\pdfoutput=0
\documentclass[conference]{IEEEtran}
\usepackage{amsfonts}
\usepackage{graphicx}
\usepackage{color}
\usepackage{amsmath,amsfonts,amssymb,amsthm,epsfig,epstopdf,url,array}
\usepackage{url,textcomp}
\usepackage{authblk}
\usepackage{cite}

\newtheorem{theorem}{Theorem}
\newtheorem{lemma}{Lemma}

\begin{document}
\title{Optimal Power Allocation for HARQ Schemes over Time-Correlated Nakagami-m Fading Channels}
\specialpapernotice{(Invited Paper)}
\author{Zheng~Shi$^1$,
        Shaodan~Ma$^1$,
        Fen Hou$^1$,
        Kam-Weng~Tam$^1$,
        and Yik-Chung Wu$^2$
        }
\affil{$^1$Department of Electrical and Computer Engineering, University of Macau, Macau \\ $^2$Department of Electrical and Electronic Engineering, The University of Hong Kong, Hong Kong}
\maketitle
\begin{abstract}
This paper investigates the problem of power allocation for hybrid automatic repeat request (HARQ) schemes over time-correlated Nakagami-m fading channels under outage constraint. The presence of time correlation complicates the power allocation problem due to the involvement of multiple correlated fading channels. Under a general time-correlated Nakagami-m fading channel with exponential correlation, outage probabilities for three widely adopted HARQ schemes, including Type I HARQ, HARQ with chase combining (HARQ-CC) and HARQ with incremental redundancy (HARQ-IR), are first derived. With these results, power allocation schemes are proposed to minimize the average total transmission power with guaranteed outage performance. Simulation results demonstrate the accuracy of our outage analysis and the effectiveness of our proposed power allocation schemes. It is shown that our proposed power allocation schemes can achieve significant power savings when compared with fixed power allocation. Moreover, under practical low outage constraint, the power efficiency is further improved when the time correlation is reduced and/or the fading order is increased.
\end{abstract}
\IEEEpeerreviewmaketitle
\hyphenation{HARQ}
\section{Introduction}\label{sec:int}
Hybrid automatic repeat request (HARQ) is a powerful transmission protocol to combat the detrimental effects of channel fading and noise due to its combination of automatic repeat request and forward error correction. Generally, there are three types of HARQ schemes, including Type I HARQ, HARQ with chase combining (HARQ-CC) and HARQ with incremental redundancy (HARQ-IR). For Type I HARQ, the erroneously received packets are discarded and only the most recently received packet is used for decoding. Since the failed packet may still contain some useful information, it can be exploited for performance enhancement and the other two HARQ schemes are thus designed for this purpose. They combine the erroneously received packets with subsequently received packets for joint decoding to improve the performance. Their difference lies in whether the same set of coded bits are transmitted in each HARQ round. Specifically, for HARQ-CC, the same coded sequence is repetitively transmitted in each HARQ round and maximal-ratio-combining (MRC) is employed to combine all the received packets to recover the message, whereas HARQ-IR transmits different sets of coded bits in each retransmission and code combining is adopted for joint decoding.

Power allocation for HARQ schemes has attracted considerable research attention recently. Most of prior works consider either quasi-static fading channels \cite{su2011optimal,makki2013green,wang2014optimum} or fast fading channels \cite{jinho2013energy,chaitanya2013optimal,chaitanya2015energy}. To be specific, in \cite{su2011optimal}, an optimal power allocation scheme is proposed to minimize the average total transmission power of HARQ-CC over quasi-static fading channels, where the channel response remains constant during multiple HARQ rounds. Similar to \cite{su2011optimal}, outage-limited power allocation is investigated for HARQ-CC and HARQ-IR schemes in both continuous and bursting communication systems in \cite{makki2013green}. Considering the same quasi-static fading channels, power allocation is investigated in \cite{wang2014optimum}. A backward sequential calculation method is developed to find the optimum power allocation. On the other hand, some of prior literature considers fast fading channels, where channel responses vary independently among multiple transmissions. For example, in \cite{jinho2013energy}, power allocation is discussed for HARQ-IR enabled distributed cooperative beamforming system, where the source and the relay have fixed transmission power in each HARQ round. Another power allocation scheme is proposed for HARQ-CC over independent Rayleigh fading channels in \cite{chaitanya2013optimal}. By reformulating the power allocation problem as a geometric programming problem and using dominant-term approximation, the optimal solution is found efficiently. The same approach is further extended to the power allocation for HARQ-enabled incremental MIMO systems in \cite{chaitanya2015energy}.

Apart from quasi-static and fast fading channels, another frequently experienced channel is time-correlated fading channel \cite{kim2011optimal,jin2011optimal}, which usually occurs when the transceiver has low-to-medium mobility. Under time correlated fading channels, power allocation becomes much more challenging due to the involvement of multiple correlated random variables and there are few solutions if any in the literature. In this paper, we investigate power allocation for HARQ schemes over time-correlated Nakagami-m fading channels. A general multivariate Nakagami-m distribution with exponential correlation is adopted to model time-correlated fading channels. The outage probabilities and their asymptotic expressions are first derived for three HARQ schemes, i.e., Type I HARQ, HARQ-CC and HARQ-IR. These analytical results then enable the optimal power allocation to minimize the average total transmission power with guaranteed outage performance. Closed-form optimal solutions are found based on the asymptotic outage expressions. Finally, these theoretical results are validated through  simulations. It is found that our proposed power allocation schemes can achieve significant power savings when compared with fixed power allocation. Moreover, under practical low outage constraint, the power efficiency is further improved with the reduction of time correlation and the increase of fading order.

The remainder of this paper is organized as follows. In Section \ref{sec:mod}, system model is given and outage analysis is conducted for three HARQ schemes. Section \ref{sec:op_cl} generalizes the problem of outage-limited power allocation for three HARQ schemes, and optimal solutions are proposed in closed forms. In Section \ref{sec:numer}, numerical results are presented and discussed to demonstrate the efficiency of our proposed power allocation schemes. Finally, Section \ref{sec:con} concludes this paper.

\section{System Model and Outage Analysis}\label{sec:mod}
A point-to-point HARQ enabled system operating over time-correlated Nakagami-m block-fading channels is considered in this paper. Following the HARQ protocol, $L$ maximal transmissions are allowed for each single message. The received signal $y_l$ in the $l$th HARQ round is written as
\begin{equation}\label{eqn:channel_mod}
  y_l = \sqrt{P_l} h_l x_l + \eta_l,\quad 0 \le l \le L,
\end{equation}
where $x_l$ denotes the transmitted signal with unit power in the $l$th HARQ round, $P_l$ refers to the transmit power in the $l$th HARQ round, $\eta_l$ represents the complex Gaussian white noise with zero mean and unit variance, i.e., $\eta_l \sim {\mathcal{CN}}(0,1)$, and $h_l$ is the channel coefficient in the $l$th HARQ round. Unlike prior literature, time correlated Nakagami-m fading channels are considered. More precisely, the joint distribution of channel amplitudes $|{\bf{h}}_L| = [|h_1|,\cdots,|h_L|]$ are modeled as a multivariate Nakagami-m distribution with exponential correlation \cite{kim2011optimal,jin2011optimal,beaulieu2011novel}, whose joint probability density function (PDF) is given by
\begin{multline}\label{eqn:joint_PDF}
{f_{|{\bf{h}}_L|}}\left( {{z_1}, \cdots ,{z_L}} \right) = \int_{t = 0}^\infty  {\frac{{{t^{m - 1}}}}{{\Gamma \left( m \right)}}{{\rm{e}}^{ - t}} } \times \\
\prod\limits_{l = 1}^L {\frac{{2{z_l}^{2m - 1}}}{{\Gamma \left( m \right){{\left( {\frac{{{\Omega _l}\left( {1 - {\rho ^{2\left( {l + \delta  - 1} \right)}}} \right)}}{m}} \right)}^m}}}{e^{ - \frac{{m{z_l}^2}}{{{\Omega _l}\left( {1 - {\rho ^{2\left( {l + \delta  - 1} \right)}}} \right)}}}}} \times\\
  {e^{ - \frac{{{\rho ^{2\left( {l + \delta  - 1} \right)}}t}}{{1 - {\rho ^{2\left( {l + \delta  - 1} \right)}}}}}}_0{F_1}\left( {;m;\frac{{m{z_l}^2{\rho ^{2\left( {l + \delta  - 1} \right)}}t}}{{{\Omega _l}{{\left( {1 - {\rho ^{2\left( {l + \delta  - 1} \right)}}} \right)}^2}}}} \right)dt, \rho \ne 1,
\end{multline}
where $\rho$ and $\delta$ denote the time correlation and the channel feedback delay, $m$ denotes the fading order, ${{\Omega _l}}$ is defined as the average power of $h_l$, i.e., ${{\Omega _l}}={\rm E}\{|h_l|^2\}$, $\Gamma(\cdot)$ denotes Gamma function and ${}_0{F_1}(\cdot)$ denotes the confluent hypergeometric limit function \cite[Eq. 9.14.1]{gradshteyn1965table}.

The system performance is fully characterized by outage probability, which is defined as the probability that the message cannot be successfully decoded, i.e., the mutual information is smaller than the target transmission rate $R$ bps/Hz. For different HARQ schemes, the outage probability over time-correlated Nakagami-m fading channels are analyzed as follows.
\subsection{Outage Probability of Type I HARQ}
For Type I HARQ, only the most recently received packet is employed for recovering the message. The outage probability $p_{out,l}^{\rm I}$ after $l$ transmissions can be formulated as
\begin{multline}\label{eqn:out_prob}
p_{out,l}^{\rm{I}} = \Pr \left( {{I_1} < R, \cdots ,{I_l} < R} \right)  \\
 = {F_{|{{\bf{h}}_l}|}}\left( {\left| {{h_1}} \right| < \sqrt {\frac{{{2^R} - 1}}{{{P_1}}}} , \cdots ,\left| {{h_l}} \right| < \sqrt {\frac{{{2^R} - 1}}{{{P_l}}}} } \right),
\end{multline}
where ${I_\iota } = {\log _2}\left( {1 + {P_\iota }{{\left| {{h_\iota }} \right|}^2}} \right)$ denotes the mutual information in the $\iota$th transmission, and ${F_{|{{\bf{h}}_l}|}} (\cdot)$ denotes the joint cumulative distribution function (CDF) with respect to $|{{\bf{h}}_l}|$, which can be derived in the following lemma.
\begin{lemma}\label{lem:joint_CDF_I} The joint CDF ${F_{|{{\bf{h}}_l}|}} (y_1,\cdots,y_l)$ can be written as a weighted sum of joint CDF of $l$ independent Nakagami RVs ${\bf A}_{\bf n}$ with parameters $(m+n_\iota,{{{\Omega _\iota }\left( {1 - {\rho ^{2\left( {\iota  + \delta  - 1} \right)}}} \right)}}(m+n_\iota)/m)$, where ${\bf n} = [n_1, \cdots, n_l]$ and $0 \le \iota \le l$. Precisely,
\begin{equation}\label{eqn:CDF_joint_exp_I}
{F_{|{{\bf{h}}_l}|}}\left( {{y_1}, \cdots ,{y_l}} \right) = \sum\limits_{{n_1}, \cdots ,{n_l } = 0}^\infty  {{W_{\bf{n}}}{F_{{{\bf{A}}_{\bf{n}}}}}\left( {{y_1}, \cdots ,{y_l}} \right)},
\end{equation}
where the coefficient $W_{\bf n}$ is given by
\begin{multline}\label{eqn:Wn_define}
{W_{\bf{n}}} = \frac{{\Gamma \left( {m + \sum\limits_{\iota  = 1}^l {{n_\iota }} } \right)}}{{\Gamma \left( m \right){{\left( {1 + \sum\limits_{\iota  = 1}^l {{\omega _\iota }} } \right)}^m}}}\prod\limits_{\iota  = 1}^l {\frac{1}{{{n_\iota }!}}{{\left( {\frac{{{\omega _\iota }}}{{1 + \sum\limits_{\iota  = 1}^l {{\omega _\iota }} }}} \right)}^{{n_\iota }}}}
\end{multline}
and satisfies $\sum\limits_{{n_1}, \cdots ,{n_l } = 0}^\infty  {{W_{\bf{n}}}}  = 1$, ${\omega _\iota }{\rm{ = }}\frac{{{\rho ^{2\left( {\iota  + \delta  - 1} \right)}}}}{{1 - {\rho ^{2\left( {\iota  + \delta  - 1} \right)}}}} $, and the joint CDF with respect to ${\bf A}_{\bf n}$, ${F_{{{\bf{A}}_{\bf{n}}}}}\left( {{y_1}, \cdots ,{y_l}} \right)$, is explicitly expressed as
\begin{equation}\label{eqn:F_A_n_joint_CDF}
{F_{{{\bf{A}}_{\bf{n}}}}}\left( {{y_1}, \cdots ,{y_l}} \right) = \prod\limits_{\iota  = 1}^l {\frac{{\Upsilon \left( {m + {n_\iota },\frac{{m{y_\iota }^2}}{{{\Omega _\iota }\left( {1 - {\rho ^{2\left( {\iota  + \delta  - 1} \right)}}} \right)}}} \right)}}{{\Gamma \left( {m + {n_\iota }} \right)}}}
\end{equation}
with $\Upsilon(\cdot,\cdot)$ being the lower incomplete Gamma function.
\end{lemma}
\begin{proof}
The result directly follows from (\ref{eqn:joint_PDF}) and the series expansion of ${}_0{F_1}(\cdot)$ \cite[Eq. 9.14.1]{gradshteyn1965table}.
\end{proof}

With Lemma \ref{lem:joint_CDF_I}, the outage probability of Type I HARQ can be obtained as
\begin{equation}\label{eqn:out_type_I_fina}
p_{out,l}^{\rm{I}} = \sum\limits_{{n_1}, \cdots ,{n_l } = 0}^\infty  {{W_{\bf{n}}}{F_{{{\bf{A}}_{\bf{n}}}}}\left( {\sqrt {\frac{{{2^R} - 1}}{{{P_1}}}} , \cdots ,\sqrt {\frac{{{2^R} - 1}}{{{P_l}}}} } \right)}.
\end{equation}
In practice, the outage probability can be computed by truncating the infinite series in (\ref{eqn:out_type_I_fina}). Herein, an efficient truncation method is proposed as
\begin{equation}\label{eqn:truncation_method}
\tilde p_{out,l}^{\rm{I}} = \sum\limits_{t = 0}^N {\sum\limits_{{n_1}{\rm{ + }} \cdots {\rm{ + }}{n_l} = t}^\infty  {{W_{\bf{n}}}{F_{{{\bf{A}}_{\bf{n}}}}}\left( {\sqrt {\frac{{{2^R} - 1}}{{{P_1}}}} , \cdots ,\sqrt {\frac{{{2^R} - 1}}{{{P_l}}}} } \right)} },
\end{equation}
where $N$ defines the truncation order. It can be proved that the truncation error exponentially decreases with $N$. The proof is omitted here due to space limit.

Under high SNR, the outage probability can be asymptotically derived as shown in following theorem.
\begin{theorem}
Under high SNR regime, i.e., $P_1,\cdots,P_l \to \infty$, the outage probability $p_{out,l}^{\rm{I}}$ is written as
\begin{equation}\label{eqn:out_prob_asym}
p_{out,l}^{\rm{I}} = \frac{{{m^{ml}}\ell \left( l,\rho  \right){{\left( {{2^R} - 1} \right)}^{lm}}}}{{{\Gamma ^l}\left( {m + 1} \right)\prod\limits_{\iota  = 1}^l {{\Omega _\iota }^m{P_\iota }^m} }},
\end{equation}
where $\ell \left( l,\rho  \right) = {\left( {\left( {1 + \sum\limits_{\iota  = 1}^l {\frac{{{\rho ^{2\left( {\iota  + \delta  - 1} \right)}}}}{{1 - {\rho ^{2\left( {\iota  + \delta  - 1} \right)}}}}} } \right)\prod\limits_{\iota  = 1}^l {\left( {1 - {\rho ^{2\left( {\iota  + \delta  - 1} \right)}}} \right)} } \right)^{ - m}}$.
\end{theorem}
\begin{proof}
By using the series expansion of $\Upsilon(\cdot,\cdot)$ \cite[Eq. 8.354.1]{gradshteyn1965table} and omitting the higher order infinitesimal of ${{{ {\prod\limits_{\iota  = 1}^l {{P_\iota }}^{ - m} } }}} $, the outage probability (\ref{eqn:out_type_I_fina}) can be asymptotically expressed as (\ref{eqn:out_prob_asym}).
\end{proof}

\subsection{Outage Probability of HARQ-CC}
In HARQ-CC scheme, all the previously received packets are combined through MRC for decoding. The outage probability after $l$ HARQ rounds is thus written as
\begin{align}\label{eqn:out_harq_cc}
p_{out,l}^{CC} &= \Pr \left( {{{\log }_2}\left( {1 + \sum\limits_{\iota  = 1}^l {{P_\iota }{{\left| {{h_\iota }} \right|}^2}} } \right) < R} \right) \notag \\
 &= \Pr \left( {Y_l \triangleq \sum\limits_{\iota  = 1}^l {{P_\iota }{{\left| {{h_\iota }} \right|}^2}}  < {2^R} - 1} \right) = {F_{{Y_l}}}\left( 2^R-1 \right).
\end{align}
where ${F_{{Y_l}}}\left( \cdot \right)$ denotes the CDF of $Y_l$. After deriving the CDF ${F_{{Y_l}}}\left( \cdot \right)$ using the method of moment generating function (MGF), the outage probability $p_{out,l}^{CC}$ is derived in the following theorem. 
\begin{theorem} \label{lem:harq_cc}
The outage probability for HARQ-CC scheme $p_{out,l}^{CC}$ can be obtained as
\begin{multline}\label{eqn:inverse_lp_cal_sec}
p_{out,l}^{CC} = 1 + \frac{{{m^{ml}}\ell \left( l,\rho  \right)}}{{\prod\nolimits_{\iota  = 1}^l {{\Omega _\iota }^m{P_\iota }^m} }} \times \\
\sum\limits_{\kappa  = 1}^{\cal K} {\sum\limits_{\varsigma  = 1}^{m{q_\kappa }} {\frac{{{\Phi _{\kappa \varsigma }}\left( { - {\lambda _\kappa }} \right)}}{{\left( {m{q_\kappa } - \varsigma } \right)!\left( {\varsigma  - 1} \right)!}}{(2^R-1)^{m{q_\kappa } - \varsigma }}{e^{ - {\lambda _\kappa }(2^R-1)}}} }
\end{multline}
where $\lambda_1,\cdots,\lambda_{\mathcal K}$ define $\mathcal K$ distinct poles of the MGF of $Y_l$ with multiplicities $q_1,\cdots,q_{\mathcal K}$, respectively, $\sum\nolimits_{\kappa  = 1}^{\mathcal K} {{q_\kappa }}  = l$, and  ${\Phi _{\kappa \varsigma }}\left( s \right) = \frac{{{d^{\varsigma  - 1}}}}{{d{s^{\varsigma  - 1}}}}\left( {{s^{ - 1}}\prod\limits_{\tau  = 1,\tau  \ne \kappa }^{\cal K} {{{\left( {s{\rm{ + }}{\lambda _\tau }} \right)}^{ - m{q_\tau }}}} } \right)$. Under high SNR regime, the outage probability $p_{out,l}^{CC}$ can also be expressed asymptotically as
\begin{equation}\label{eqn:out_cc_asy}
p_{out,l}^{CC} = \frac{{{m^{ml}}\ell \left( l,\rho  \right){{\left( {{2^R} - 1} \right)}^{ml}}}}{{\Gamma \left( {ml + 1} \right)\prod\nolimits_{\iota  = 1}^l {{\Omega _\iota }^m{P_\iota }^m} }}.
\end{equation}
\end{theorem}
\begin{proof}
Please see Appendix \ref{app:harq_cc}.
\end{proof}
\subsection{Outage Probability of HARQ-IR}
Different from Type I HARQ and HARQ-CC, HARQ-IR accumulates mutual information in all previous HARQ rounds for decoding. From information theoretical perspective, an outage happens when the accumulated mutual information is less than the target transmission rate $R$. Thus the outage probability after $l$ HARQ rounds is formulated as
\begin{equation}\label{eqn:outage_prob_IR}
p_{out,l}^{IR} = \Pr \left( {\sum\limits_{\iota  = 1}^l {{{\log }_2}\left( {1 + {P_\iota }{{\left| {{h_\iota }} \right|}^2}} \right)}  < R} \right).
\end{equation}
Due to the time correlation among $h_l$, it is intractable to find closed-form expression for (\ref{eqn:outage_prob_IR}). Instead, a lower bound of $p_{out,l}^{IR}$ is adopted to characterize the outage probability of HARQ-IR. By using Jensen's inequality, $p_{out,l}^{IR}$ is lower bounded as
\begin{align}\label{eqn:ir_uper_bound}
p_{out,l}^{IR} &\ge \Pr \left( {{{\log }_2}\left( {1 + \frac{1}{l}\sum\limits_{\iota  = 1}^l {{P_\iota }{{\left| {{h_\iota }} \right|}^2}} } \right) < \frac{R}{l}} \right) \notag \\
&= {F_{{Y_l}}}\left( {l\left( {{2^{R/l}} - 1} \right)} \right) \triangleq {p_{out,l}^{IR,lower}}.
\end{align}
With the CDF ${F_{{Y_l}}}\left( \cdot \right)$ derived in Theorem \ref{lem:harq_cc}, the lower bound ${p_{out,l}^{IR,lower}}$ and its asymptotic expression can be derived in the following theorem.
\begin{theorem} \label{lem:harq_ir}
The lower bound of the outage probability $p_{out,l}^{IR,lower}$ can be obtained as
\begin{multline}\label{eqn:inverse_lp_cal_sec_ir}
{p_{out,l}^{IR,lower}} = 1 + \frac{{{m^{ml}}\ell \left( l,\rho  \right)}}{{\prod\nolimits_{\iota  = 1}^l {{\Omega _\iota }^m{P_\iota }^m} }} \times \\
\sum\limits_{\kappa  = 1}^{\cal K} {\sum\limits_{\varsigma  = 1}^{m{q_\kappa }} {\frac{{{\Phi _{\kappa \varsigma }}\left( { - {\lambda _\kappa }} \right)}}{{\left( {m{q_\kappa } - \varsigma } \right)!\left( {\varsigma  - 1} \right)!}}{(l({{2^{R/l}} - 1}))^{m{q_\kappa } - \varsigma }}{e^{ - {\lambda _\kappa }(l({{2^{R/l}} - 1}))}}} }.
\end{multline}
Under high SNR regime, ${p_{out,l}^{IR,lower}}$ is further simplified as
\begin{equation}\label{eqn:out_cc_asy_ir}
p_{out,l}^{IR,lower} = \frac{{{(lm)^{ml}}\ell \left( l,\rho  \right){{\left( {{2^{R/l}} - 1} \right)}^{ml}}}}{{\Gamma \left( {ml + 1} \right)\prod\nolimits_{\iota  = 1}^l {{\Omega _\iota }^m{P_\iota }^m} }}.
\end{equation}
\end{theorem}

\section{Optimal Power allocation}\label{sec:op_cl}
In this section, the problem of power allocation is studied for the three HARQ schemes. Generally, the average total transmission power for HARQ is defined as ${\bar P = \sum\nolimits_{l = 1}^L {{P_l}{p_{out,l - 1}}} }$ \cite{chaitanya2013optimal}. Here ${p_{out,l}}$ refers to the outage probability after $l$ transmissions and it unifies the cases of ${p_{out,l}^{I}}$, ${p_{out,l}^{CC}}$ and ${p_{out,l}^{IR,lower}}$. When power efficiency is concerned with certain performance requirement, the transmission power among multiple HARQ rounds should be properly designed to minimize the total transmission power while guaranteeing the performance. The power allocation problem can be formulated as

\begin{equation}\label{eqn:opt_prob_simp}
\begin{array}{*{20}{cl}}
{\mathop {\min }\limits_{{P_1},{P_2}, \cdots {P_L}} }&{\bar P = \sum\limits_{l = 1}^L {{P_l}{p_{out,l - 1}}} }\\
{{\rm{s}}{\rm{.t}}{\rm{.}}}&{{P_l} \ge 0,0 \le l \le L}\\
{}&{{p_{out,L}} \le \varepsilon }\\
\end{array},
\end{equation}
where $\varepsilon$ represents the outage tolerance.

Due to the complicated expressions of the exact outage probabilities given in (\ref{eqn:out_type_I_fina}), (\ref{eqn:inverse_lp_cal_sec}) and (\ref{eqn:inverse_lp_cal_sec_ir}), it is impossible to find closed-form power allocation solutions directly. However, interior-point methods can be exploited to numerically solve the problem (\ref{eqn:opt_prob_simp}). Meanwhile, based on the asymptotic expressions of the outage probabilities, an efficient power allocation scheme can be found as follows.

Notice that the asymptotic outage probabilities in (\ref{eqn:out_prob_asym}), (\ref{eqn:out_cc_asy}) and (\ref{eqn:out_cc_asy_ir}) can be unified as
\begin{equation}\label{eqn:Pout_def}
{p_{out,l}} \simeq  \frac{{{\phi _l}}}{{{{\left( {\prod\limits_{k = 1}^l {{P_k}} } \right)}^m}}},\, 0 \le l \le L,
\end{equation}
where $\phi_l$ depends on HARQ schemes, more precisely,
\begin{equation}\label{eqn:Wl_def}
{\phi_l} = \left\{ {\begin{array}{*{20}{cl}}
{\frac{{{m^{ml}}\ell \left( {l,\rho } \right){{\left( {{2^R} - 1} \right)}^{ml}}}}{{{\Gamma ^l}\left( {m + 1} \right)\prod\nolimits_{\iota  = 1}^l {{\Omega _\iota }^m} }},}&{{\textrm{Type}}\;{\textrm{I}};}\\
{\frac{{{m^{ml}}\ell \left( {l,\rho } \right){{\left( {{2^R} - 1} \right)}^{ml}}}}{{\Gamma \left( {ml + 1} \right)\prod\nolimits_{\iota  = 1}^l {{\Omega _\iota }^m} }},}&{{\textrm{HARQ-CC}};}\\
{\frac{{{{\left( {ml} \right)}^{ml}}\ell \left( {l,\rho } \right){{\left( {{2^{R/l}} - 1} \right)}^{ml}}}}{{\Gamma \left( {ml + 1} \right)\prod\nolimits_{\iota  = 1}^l {{\Omega _\iota }^m} }},}&{{\textrm{HARQ-IR}}.}
\end{array}} \right.
\end{equation}

Substituting (\ref{eqn:Pout_def}) into (\ref{eqn:opt_prob_simp}), the Lagrangian of the optimization problem (\ref{eqn:opt_prob_simp}) is formed as
\begin{multline}\label{eqn:lagrangian_power}
\frak L\left( {{P_1}, \cdots ,{P_L},\mu ,{\nu _1}, \cdots ,{\nu _L}} \right) = \sum\limits_{l = 1}^L {{P_l}\frac{{{\phi_{l - 1}}}}{{{{\left( {\prod\limits_{k = 1}^{l - 1} {{P_k}} } \right)}^m}}}} \\
 +  \mu \left( {\frac{{{\phi_L}}}{{{{\left( {\prod\limits_{k = 1}^L {{P_k}} } \right)}^m}}} - \varepsilon } \right) - \sum\limits_{l = 1}^L {{\nu _l}{P_l}},
\end{multline}
where $\mu,\nu_1,\cdots,\nu_L$ are the Lagrangian multipliers of the constraints in the problem (\ref{eqn:opt_prob_simp}). Since the Karush-Khun-Tucker (KKT) conditions are necessary for an optimal solution, we have
\begin{equation}\label{eqn:kkt1}
{\left. {\frac{{\partial \frak L}}{{\partial {P_n}}}} \right|_{\left( {{P_1^*}, \cdots ,{P_L^*},{\mu ^*},{\nu _1}^*, \cdots ,{\nu _L}^*} \right)}} = 0,
\end{equation}
\begin{equation}\label{eqn:kkt2}
{\mu ^*}\left( {\frac{{{\phi_L}}}{{{{\left( {\prod\limits_{k = 1}^L {P_k^*} } \right)}^m}}} - \varepsilon } \right) = 0,
\end{equation}
\begin{equation}\label{eqn:kkt3}
{\nu _l}^*{P_l}^* = 0,
\end{equation}
where $\mu^*,{\nu_1}^*,\cdots,{\nu_L}^*, {P_l}^*$ denote the optimal Lagrangian multipliers and the optimal power allocation, respectively. Based on the KKT conditions (\ref{eqn:kkt1})-(\ref{eqn:kkt3}), the optimal power allocation solution to (\ref{eqn:opt_prob_simp}) could be found in closed form as follows.
\begin{theorem}\label{the:power_opt}
The optimal solution to the problem (\ref{eqn:opt_prob_simp}) is uniquely given as
\begin{align}\label{eqn:P_L_star_sec}
P_L^* 
 &= {\left( {\frac{{{\phi_L}\prod\limits_{k = 2}^L {{{\left( {\left( {m + 1} \right)\frac{{{\phi_{k - 1}}}}{{{\phi_{k - 2}}}}} \right)}^{\frac{1}{{{{\left( {m + 1} \right)}^{k - 1}}}}}}} }}{{{\phi_{L - 1}}{{\left( {m + 1} \right)}^{L - 1}}\varepsilon }}} \right)^{\frac{{{{\left( {m + 1} \right)}^{L - 1}}}}{{{{\left( {m + 1} \right)}^L} - 1}}}},
\end{align}
\begin{multline}\label{eqn:kkt1_subs_fin}
P_n^* = \prod\limits_{k = n + 1}^L {{{\left( {\left( {m + 1} \right)\frac{{{\phi_{k - 1}}}}{{{\phi_{k - 2}}}}} \right)}^{\frac{1}{{{{\left( {m + 1} \right)}^{k - n}}}}}}} {P_L^*}^{\frac{1}{{{{\left( {m + 1} \right)}^{L - n}}}}}, \\
\quad {\textrm for}\quad 1 \le n \le L-1.
\end{multline}
Moreover, the minimal average total transmission power $\bar P^*$ is
\begin{align}\label{eqn:max_avg_pow_sec}
\bar P^* &= \frac{{\varepsilon \left({P_L}{^*}\right)^{ {m + 1} }{\phi_{L - 1}}}}{{{\phi_L}}}\frac{{ {{{\left( {m + 1} \right)}^L} - 1} }}{m}.
\end{align}
\end{theorem}
\begin{proof}
Please see Appendix \ref{app:power_opt}.
\end{proof}

\section{Numerical Results and Discussions}\label{sec:numer}
In this section, the analytical results are verified through simulations, and our proposed power allocation (PPA) scheme is compared to the fixed power allocation (FPA) scheme in \cite{jinho2013energy}. Notice that for FPA, the problem (\ref{eqn:opt_prob_simp}) is solved by adding the constraint $P_1=\cdots=P_L$. In the sequel, we take systems with $\Omega_1=\cdots=\Omega_l=1$, $\delta=1$s and $R=2$bps/Hz as examples.

\subsection{Comparison of PPA and FPA}
The minimal total transmission powers $\bar P^*$ of the PPA and FPA schemes are compared in Fig. \ref{fig:val}. The outage-limited systems with various HARQ and parameters as $L=2$, $m=2$ and $\rho=0.5$ are considered. Clearly, the results of PPA using asymptotic outage expressions (PPA-A) agree well with that of PPA using exact outage expressions (PPA-E) under low outage constraint $\varepsilon$. It is also readily found that PPA-A performs better than FPA under low outage constraint $\varepsilon$, and their performance gap significantly increases when $\varepsilon$ decreases. Moreover, it reveals that HARQ-IR is superior to Type I HARQ and HARQ-CC in terms of power efficiency.
\begin{figure}
  \centering
  \includegraphics[width=3in]{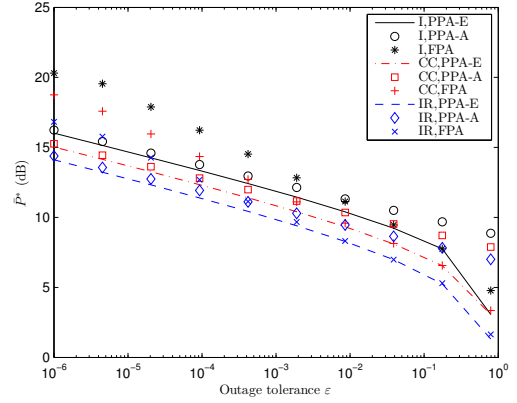}\\
  \caption{Comparison of the proposed power allocation with the fixed power allocation.}\label{fig:val}
\end{figure}

\subsection{Impacts of Time Correlation}
Since the performance of PPA-A is asymptotically equivalent to that of PPA-E under low $\varepsilon$, PPA-A is adopted to test the impact of time correlation on power allocation. Fig. \ref{fig:imp_corr} plots the minimal total transmission power $\bar P^*$ against time correlation coefficient $\rho$ by setting parameters as $m=2$ and $\varepsilon=10^{-6}$. It can be easily seen that the increase of time correlation $\rho$ would lead to the increase of the minimal total transmission power $\bar P^*$ for both $L=2$ and $L=4$. It means that time correlation has negative effect on power efficiency under low outage constraint.
\begin{figure}
  \centering
  \includegraphics[width=3in]{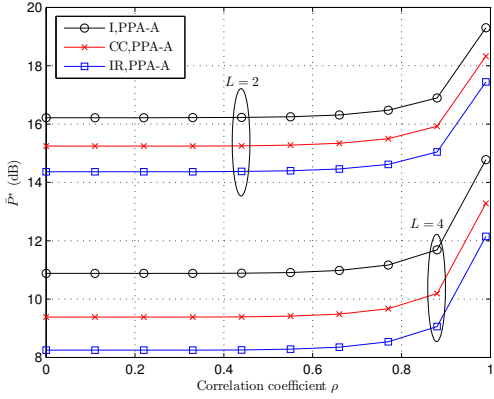}\\
  \caption{Impact of time correlation.}\label{fig:imp_corr}
\end{figure}
\subsection{Impacts of Fading Order}
Fig. \ref{fig:order} depicts the impact of fading order $m$ on power allocation by setting $\rho=0.5$ and $\varepsilon=10^{-6}$. Clearly, the minimal total transmission power $\bar P^*$ decreases with the increase of fading order $m$. In fact, higher fading order leads to higher diversity introduced by the channel, thus reducing the power consumption given a certain outage performance constraint.
\begin{figure}
  \centering
  \includegraphics[width=3in]{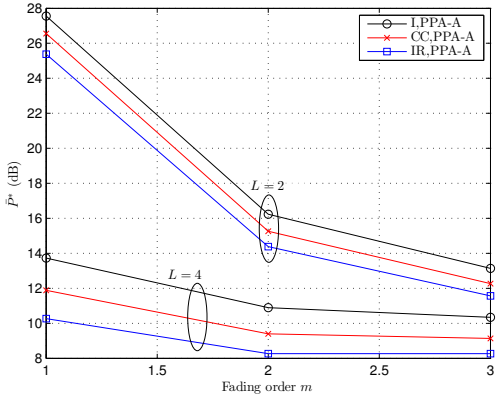}\\
  \caption{Impact of fading order.}\label{fig:order}
\end{figure}

\section{Conclusions}\label{sec:con}
Outage-limited power allocation for various HARQ schemes operating over time-correlated Nakagami-m fading channels has been investigated in this paper. After deriving the outage probabilities and their asymptotic expressions, an optimal power allocation solution has been proposed in closed form. It has been demonstrated that our proposed solution can achieve significant power saving when compared to the fixed power allocation solution. The superiority of the proposed optimal solution in terms of power consumption is further enhanced when the channel time correlation is reduced and/or the fading order increases.
\section{Acknowledgements}
This work was supported in part by the Research Committee of University of Macau under grants: MYRG101(Y1-L3)-FST13-MSD, MYRG2014-00146-FST and MYRG2016-00146-FST, in part by the Macau Science and Technology Development Fund under grant FDCT 091/2015/A3, and in part by the National Natural Science Foundation of China under grant No.61601524.
\appendices
\section{Proof of Theorem \ref{lem:harq_cc}}\label{app:harq_cc}
The moment generation function (MGF) with respect to $Y_l$ can be written as
\begin{multline}\label{eqn:mgf_Y_l}
{M_{{Y_l}}}\left( s \right) = {\rm{E}}\left( {{e^{s{Y_l}}}} \right)\\
 = \int\limits_0^\infty  { \cdots \int\limits_0^\infty  {{e^{s\sum\limits_{\iota  = 1}^l {{P_\iota }{x_\iota }^2} }}{f_{|{{\bf{h}}_l}|}}\left( {{x_1}, \cdots ,{x_l}} \right)d{x_1} \cdots {dx_l}} }.
\end{multline}
Plugging (\ref{eqn:joint_PDF}) into (\ref{eqn:mgf_Y_l}), after some algebraic manipulations, it follows that
\begin{align}\label{eqn:mgf_Y_l_2}
{M_{{Y_l}}}\left( s \right) &= {{{\left( \begin{array}{l}
\prod\limits_{\iota  = 1}^l {\left( {1 - s{P_\iota }{\Omega _\iota }\left( {1 - {\rho ^{2\left( {\iota  + \delta  - 1} \right)}}} \right)/m} \right)} \\
 \times \left( {1 + \sum\limits_{\iota  = 1}^l {\frac{{s{P_\iota }{\Omega _\iota }{\rho ^{2\left( {\iota  + \delta  - 1} \right)}}/m}}{{s{P_\iota }{\Omega _\iota }\left( {1 - {\rho ^{2\left( {\iota  + \delta  - 1} \right)}}} \right)/m - 1}}} } \right)
\end{array} \right)}^{ - m}}} \notag \\
& = \frac{{{m^{ml}}\ell \left( l,\rho  \right)}}{{\prod\limits_{\iota  = 1}^l {{\Omega _\iota }^m{P_\iota }^m} }}\prod\limits_{\kappa  = 1}^{\cal K} {{{\left( {{\lambda _\kappa } - s} \right)}^{ - m{q_\kappa }}}},
\end{align}
where $\lambda_1,\cdots,\lambda_{\mathcal K}$ define $\mathcal K$ distinct poles of ${M_{{Y_l}}}\left( s \right)$ with multiplicities $q_1,\cdots,q_{\mathcal K}$, respectively, and $\sum\limits_{\kappa  = 1}^{\mathcal K} {{q_\kappa }}  = l$. After some tedious manipulations, ${M_{{Y_l}}}\left( s \right)$ can be simplified as ${M_{{Y_l}}}\left( s \right) = \det {\left( {{\bf{I}} - s{{\bf{F}}^{1/2}{\bf{E}}{\bf{F}}^{1/2}}} \right)^{ - m}}$, where the notation $\det(\cdot)$ refers to the determinant, $\bf I$ represents an identity matrix, $\bf F$ is an $l\times l$ diagonal matrix with diagonal entries as $\{\Omega_\iota P_\iota/m\}_{\iota=1}^l$, and $\bf E$ is an $l\times l$ positive definite matrix given by
\begin{equation}\label{eqn:def_E}
{\bf{E}} = \left[ {\begin{array}{*{20}{c}}
1&{{\rho ^{2\delta  + 1}}}& \cdots &{{\rho ^{2\delta  + l - 1}}}\\
{{\rho ^{2\delta  + 1}}}&1& \cdots &{{\rho ^{2\delta  + l}}}\\
 \vdots & \vdots & \ddots & \vdots \\
{{\rho ^{2\delta  + l - 1}}}&{{\rho ^{2\delta  + l}}}& \cdots &1
\end{array}} \right].
\end{equation}
Since $1/\lambda_1,\cdots,1/\lambda_{\mathcal K}$ are the eigenvalues of the positive definite matrix ${{\bf{F}}^{1/2}{\bf{E}}{\bf{F}}^{1/2}}$, we have $\lambda_1,\cdots,\lambda_{\mathcal K} > 0$. By using inverse Laplace transform and its integration property \cite[Eq.9.109]{Oppenheim1996signals}, the CDF with respect to $Y_l$ is derived as
\begin{align}\label{eqn:inverse_lp_cdf}
{F_{{Y_l}}}\left( y \right) = \mathcal L^{-1}\left\{ {{M_{{Y_l}}}\left( { - s} \right)} \right\}\left( y \right) = \frac{1}{{2\pi j}}\int\limits_{a - j\infty }^{a + j\infty } {\frac{{{M_{{Y_l}}}\left( { - s} \right)}}{s}{e^{sy}}ds},
\end{align}
By using \cite[Eq. 5.2.21]{bateman1954tables}, (\ref{eqn:inverse_lp_cdf}) can be calculated as
\begin{multline}\label{eqn:inverse_lp_cal}
{F_{{Y_l}}}\left( y \right) = \frac{{{m^{ml}}\ell \left( l, \rho  \right)}}{{\prod\limits_{\iota  = 1}^l {{\Omega _\iota }^m{P_\iota }^m} }}\frac{1}{{2\pi j}}\int\limits_{a - j\infty }^{a + j\infty } {\frac{{{e^{sy}}}}{{s\prod\limits_{\kappa  = 1}^{\cal K} {{{\left( {s{\rm{ + }}{\lambda _\kappa }} \right)}^{m{q_\kappa }}}} }}ds}  \\
= 1 + \frac{{{m^{ml}}\ell \left( l,\rho  \right)}}{{\prod\limits_{\iota  = 1}^l {{\Omega _\iota }^m{P_\iota }^m} }}\sum\limits_{\kappa  = 1}^{\cal K} {\sum\limits_{\varsigma  = 1}^{m{q_\kappa }} {\frac{{{\Phi _{\kappa \varsigma }}\left( { - {\lambda _\kappa }} \right)}}{{\left( {m{q_\kappa } - \varsigma } \right)!\left( {\varsigma  - 1} \right)!}}{y^{m{q_\kappa } - \varsigma }}{e^{ - {\lambda _\kappa }y}}} },
\end{multline}
where ${\Phi _{\kappa \varsigma }}\left( s \right) = \frac{{{d^{\varsigma  - 1}}}}{{d{s^{\varsigma  - 1}}}}\left( {{s^{ - 1}}\prod\limits_{\tau  = 1,\tau  \ne \kappa }^{\cal K} {{{\left( {s{\rm{ + }}{\lambda _\tau }} \right)}^{ - m{q_\tau }}}} } \right)$. Therefore, by using (\ref{eqn:inverse_lp_cal}) and together with $y=2^R-1$, (\ref{eqn:inverse_lp_cal_sec}) in Theorem \ref{lem:harq_cc} is proved.

By using the expansion of Maclaurin series, (\ref{eqn:inverse_lp_cal}) can be further rewritten as
\begin{equation}\label{eqn:out_taylor}
{F_{{Y_l}}}\left( y \right) = \sum\limits_{n = 0}^\infty  {\frac{{{F_{{Y_l}}}^{\left( n \right)}\left( 0 \right)}}{{n!}}{y^n}}.
\end{equation}
Since
\begin{multline}\label{eqn:taylor_just}
{F_{{Y_l}}}^{\left( n \right)}\left( 0 \right) = \frac{{{m^{ml}}\ell \left( l,\rho  \right)}}{{\prod\limits_{\iota  = 1}^l {{\Omega _\iota }^m{P_\iota }^m} }}\frac{1}{{2\pi j}}\int\limits_{a - j\infty }^{a + j\infty } {\frac{{{s^{n - 1}}}}{{\prod\limits_{\kappa  = 1}^{\cal K} {{{\left( {s{\rm{ + }}{\lambda _\kappa }} \right)}^{m{q_\kappa }}}} }}ds}, \\
0 \le n \le ml,
\end{multline}
it follows by using the initial-value theorems of Laplace transform \cite[Eq. 9.5.10]{Oppenheim1996signals} that ${F_{{Y_l}}}^{\left( 1 \right)}\left( 0 \right)=\cdots={F_{{Y_l}}}^{\left( ml-1 \right)}\left( 0 \right)=0$, and ${F_{{Y_l}}}^{\left( ml \right)}\left( 0 \right) = \frac{{{m^{ml}}\ell \left( l,\rho  \right)}}{{\prod\nolimits_{\iota  = 1}^l {{\Omega _\iota }^m{P_\iota }^m} }}$. Moreover, it can be proved that ${F_{{Y_l}}}^{\left( n \right)}\left( 0 \right)$ is a higher order term of $\prod\nolimits_{\iota  = 1}^l {{P_\iota }^{-m}} $ for $n > ml$. The proof is omitted due to space limit. Thus it yields
\begin{equation}\label{eqn:higher_order_cc}
{F_{{Y_l}}}\left( y \right) = \frac{{{m^{ml}}\ell \left( l,\rho  \right){y^{ml}}}}{{\Gamma \left( {ml + 1} \right)\prod\limits_{\iota  = 1}^l {{\Omega _\iota }^m{P_\iota }^m} }} + o\left( {\prod\limits_{\iota  = 1}^l {{P_\iota }^{ - m}} } \right).
\end{equation}
Hence under high SNR, i.e., $P_\iota \to \infty$, (\ref{eqn:out_cc_asy}) can be derived by using (\ref{eqn:higher_order_cc}) together with $y=2^R-1$. Thus Theorem \ref{lem:harq_cc} is proved.

\section{Proof of Theorem \ref{the:power_opt}}\label{app:power_opt}
Clearly from (\ref{eqn:Pout_def}), since $P_l^* \ne 0$, we have ${\nu _l} = 0$. Therefore, after some algebraic manipulations, (\ref{eqn:kkt1}) can be rewritten as
\begin{multline}\label{eqn:kkt1_fur}
{\left. {\frac{{\partial \frak L}}{{\partial {P_n}}}} \right|_{\left( {{P_1^*}, \cdots ,{P_L^*},{\mu ^*}} \right)}} = \frac{{{\phi_{n - 1}}}}{{{{\left( {\prod\limits_{k = 1}^{n - 1} {P_k^*} } \right)}^m}}}  \\
- \frac{m}{{{P_n^*}}}\sum\limits_{l = n + 1}^L {{P_l}^*\frac{{{\phi_{l - 1}}}}{{{{\left( {\prod\limits_{k = 1}^{l - 1} {P_k^*} } \right)}^m}}}}  - \frac{m}{{{P_n^*}}}{\mu ^*}\frac{{{\phi_L}}}{{{{\left( {\prod\limits_{k = 1}^L {P_k^*} } \right)}^m}}} = 0.
\end{multline}
Together with ${\left. {\frac{{\partial \frak L}}{{\partial {P_{n-1}}}}} \right|_{\left( {{P_1^*}, \cdots ,{P_L^*},{\mu ^*}} \right)}}=0$, (\ref{eqn:kkt1_fur}) reduces to
\begin{equation}\label{eqn:kkt1_subs}
P_{n}^* = \left( {m + 1} \right){P_{n+1}^*}\frac{{{\phi_{n}}}}{{{\phi_{n - 1}}{{\left( {P_{n }^*} \right)}^m}}}.
\end{equation}
Now from (\ref{eqn:kkt1_subs}), we can derive $P_n^*$ recursively as (\ref{eqn:kkt1_subs_fin}).
Regarding to $P_L^*$, by letting $n=L$ in (\ref{eqn:kkt1_fur}), we have
\begin{equation}\label{eqn:kkt1_fur_fur}
\frac{{{\phi_{L - 1}}}}{{{{\left( {\prod\limits_{k = 1}^{L - 1} {P_k^*} } \right)}^m}}} = m{\mu ^*}\frac{{{\phi_L}}}{{{P_L^*}{{\left( {\prod\limits_{k = 1}^L {P_k^*} } \right)}^m}}} \Rightarrow {\mu ^*}{\rm{ = }}\frac{{{P_L^*}^{m + 1}{\phi_{L - 1}}}}{{m{\phi_L}}}.
\end{equation}
Recalling that $P_L \ne 0$, thus ${\mu ^*} \ne 0$. According to (\ref{eqn:kkt2}), we have
\begin{equation}\label{eqn:KKT condtion_outage}
\frac{{{\phi_L}}}{{{{\left( {\prod\limits_{k = 1}^L {P_k^*} } \right)}^m}}} - \varepsilon  = 0 \Rightarrow {\left( {\prod\limits_{n = 1}^L {{P_n^*}} } \right)^m} = \frac{{{\phi_L}}}{\varepsilon }.
\end{equation}
Substituting (\ref{eqn:kkt1_subs_fin}) into (\ref{eqn:KKT condtion_outage}) yields (\ref{eqn:P_L_star_sec}).
Moreover, by using (\ref{eqn:kkt1_subs}), it follows that
\begin{equation}\label{eqn:pn_recur_rela}
 {P_n^*} = \frac{{{{\left( {m + 1} \right)}^{L - n}}{\phi_{L - 1}}}}{{{\phi_{n - 1}}}}\frac{{{P_L^*}}}{{{{\left( {{P_n^*} \cdots P_{L - 1}^*} \right)}^m}}}.
\end{equation}

\bibliographystyle{ieeetran}
\bibliography{manuscript_1}

\begin{thebibliography}{10}
\providecommand{\url}[1]{#1}
\csname url@samestyle\endcsname
\providecommand{\newblock}{\relax}
\providecommand{\bibinfo}[2]{#2}
\providecommand{\BIBentrySTDinterwordspacing}{\spaceskip=0pt\relax}
\providecommand{\BIBentryALTinterwordstretchfactor}{4}
\providecommand{\BIBentryALTinterwordspacing}{\spaceskip=\fontdimen2\font plus
\BIBentryALTinterwordstretchfactor\fontdimen3\font minus
  \fontdimen4\font\relax}
\providecommand{\BIBforeignlanguage}[2]{{%
\expandafter\ifx\csname l@#1\endcsname\relax
\typeout{** WARNING: IEEEtran.bst: No hyphenation pattern has been}%
\typeout{** loaded for the language `#1'. Using the pattern for}%
\typeout{** the default language instead.}%
\else
\language=\csname l@#1\endcsname
\fi
#2}}
\providecommand{\BIBdecl}{\relax}
\BIBdecl

\bibitem{su2011optimal}
W.~Su, S.~Lee, D.~Pados, J.~D. Matyjas \emph{et~al.}, ``Optimal power
  assignment for minimizing the average total transmission power in {hybrid-ARQ
  Rayleigh} fading links,'' \emph{{IEEE} Trans. Commun.}, vol.~59, no.~7, pp.
  1867--1877, Jul. 2011.

\bibitem{makki2013green}
B.~Makki, A.~Graell~i Amat, and T.~Eriksson, ``Green communication via
  power-optimized {HARQ} protocols,'' \emph{{IEEE} Trans. Veh. Technol.},
  vol.~63, no.~1, pp. 161--177, Jan. 2013.

\bibitem{wang2014optimum}
G.~Wang, J.~Wu, and Y.~R. Zheng, ``Optimum energy and spectral efficient
  transmissions for delay-constrained hybrid {ARQ} systems,'' \emph{{IEEE}
  Trans. Veh. Technol.}, vol.~PP, no.~99, pp. 1--1, Aug. 2014.

\bibitem{jinho2013energy}
J.~Choi, W.~Xing, D.~To, Y.~Wu, and S.~Xu, ``On the energy efficiency of a
  relaying protocol with {HARQ-IR} and distributed cooperative beamforming,''
  \emph{{IEEE} Trans. Wireless Commun.}, vol.~12, no.~2, pp. 769--781, Feb.
  2013.

\bibitem{chaitanya2013optimal}
T.~V. Chaitanya and E.~G. Larsson, ``Optimal power allocation for hybrid {ARQ}
  with chase combining in i.i.d. {Rayleigh} fading channels,'' \emph{{IEEE}
  Trans. Commun.}, vol.~61, no.~5, pp. 1835--1846, May 2013.

\bibitem{chaitanya2015energy}
T.~Chaitanya and T.~Le-Ngoc, ``Energy-efficient adaptive power allocation for
  incremental {MIMO} systems,'' \emph{{IEEE} Trans. Veh. Technol.}, vol.~PP,
  no.~99, pp. 1--1, Mar. 2015.

\bibitem{kim2011optimal}
S.~M. Kim, W.~Choi, T.~W. Ban, and D.~K. Sung, ``Optimal rate adaptation for
  hybrid {ARQ} in time-correlated {Rayleigh} fading channels,'' \emph{{IEEE}
  Trans. Wireless Commun.}, vol.~10, no.~3, pp. 968--979, Mar. 2011.

\bibitem{jin2011optimal}
H.~Jin, C.~Cho, N.-O. Song, and D.~K. Sung, ``Optimal rate selection for
  persistent scheduling with {HARQ} in time-correlated {Nakagami-m} fading
  channels,'' \emph{{IEEE} Trans. Wireless Commun.}, vol.~10, no.~2, pp.
  637--647, Feb. 2011.

\bibitem{beaulieu2011novel}
N.~C. Beaulieu and K.~T. Hemachandra, ``Novel simple representations for
  {Gaussian} class multivariate distributions with generalized correlation,''
  \emph{{IEEE} Trans. Inf. Theory}, vol.~57, no.~12, pp. 8072--8083, Dec. 2011.

\bibitem{gradshteyn1965table}
I.~S. Gradshteyn, I.~M. Ryzhik, A.~Jeffrey, D.~Zwillinger, and S.~Technica,
  \emph{Table of integrals, series, and products}.\hskip 1em plus 0.5em minus
  0.4em\relax Academic press New York, 1965, vol.~6.

\bibitem{Oppenheim1996signals}
A.~V. Oppenheim, A.~S. Willsky, and S.~H. Nawab, \emph{Signals \&Amp; Systems
  (2Nd Ed.)}.\hskip 1em plus 0.5em minus 0.4em\relax Upper Saddle River, NJ,
  USA: Prentice-Hall, Inc., 1996.

\bibitem{bateman1954tables}
H.~Bateman, ``Tables of integral transforms,'' \emph{California Institute of
  Technology Bateman Manuscript Project, New York: McGraw-Hill, 1954, edited by
  Erdelyi, Arthur}, vol.~1, 1954.

\end{thebibliography}

\end{document}